\documentclass[12pt]{article}

\usepackage{amstext,amsmath,amssymb,amsfonts,bbm,slashed}
\usepackage[latin1]{inputenc}
\usepackage{epsfig}
\usepackage{hyperref}
\hypersetup{
    colorlinks=false,
    pdfborder={0 0 0},
}
\usepackage{xcolor}

\usepackage{amsthm}
\usepackage{color}
\usepackage{subfig}
\usepackage{cleveref}
\usepackage{comment}

\newcommand{\mtr}[1]{\mathrm{#1}}

\newcommand{\inter}{\text{int\,}}
\newcommand{\wl}{\widetilde{\lambda}}

\def\Z{{\mathbbm Z}}

\newtheorem{lemma}{Lemma}

\newcommand{\bea}{\begin{eqnarray}}
\newcommand{\eea}{\end{eqnarray}}
\newcommand{\beq}{\begin{equation}}
\newcommand{\eeq}{\end{equation}}

\newcommand{\Fm}{F_{\text{max}}}
\newcommand{\vp}{\varphi}
\newcommand{\bvp}{\bar{\varphi}}

\newcommand{\cG}{\mathcal{G}}

\makeatletter

\newcommand{\les}{\leqslant}
\newcommand{\ges}{\geqslant}

\newcommand{\set}[1]{\lb #1\rb}

\def\lb{\left \{}
  \def\rb{\right \}}

\newcommand{\rc}{\mathrm{r}}
\newcommand{\si}{\sigma}
\newcommand{\wsi}{\widetilde{\sigma}}

\begin{document}
\mdseries
\upshape
\begin{center}

{\LARGE\bf Correlation functions of just renormalizable tensorial group field theory:
 The melonic approximation }
\vspace{15pt}

{\large   Dine Ousmane Samary$^{1,4}$,    Carlos I. P\'erez-S\'anchez$^{2}$\\ Fabien Vignes-Tourneret$^3$ 
and Raimar Wulkenhaar$^{2}$}

\vspace{15pt}

{\sl 1)\,\,Perimeter Institute for Theoretical Physics\\
31 Caroline St. N.
Waterloo, ON N2L 2Y5, Canada\\

\vspace{0.5cm}
2)\,\, Mathematisches Institut der Westf\"alischen Wilhelms-Universit\"at\\
Einsteinstra\ss{}e 62, D-48149 M\"unster, Germany

\vspace{0.5cm}
3)\,\, Institut Camille Jordan, Universit\'e de Lyon, CNRS UMR 5208, \\
Villeurbanne Cedex, France
\vspace{0.5cm}

4)\,\, Facult\'e des Sciences et Techniques/  ICMPA-UNESCO Chair,\\ Universit\'e d'Abomey-Calavi, 072 BP 50, Benin\\
}
\vspace{5pt}
E-mails:  {\sl dsamary@perimeterinstitute.ca\\
c$\_$pere03@uni-muenster.de\\
 vignes@math.univ-lyon1.fr\\
raimar@math.uni-muenster.de}

\vspace{10pt}

\begin{abstract}
  The $D$-colored version of tensor models has been shown to admit a
  large $N$-limit expansion. The leading contributions result from
  so-called melonic graphs which are dual to the $D$-sphere. This is a
  note about the Schwinger-Dyson equations of the tensorial
  $\vp^{4}_{5}$-model (with propagator $1/{\bf p}^{2}$) and their
  melonic approximation. We derive the master equations for two- and
  four-point correlation functions and discuss their solution. 
\end{abstract}

\end{center}

\noindent  Pacs numbers:  11.10.Gh, 04.60.-m
\\

\setcounter{footnote}{0}



\section{Introduction} 
The construction of a consistent quantum theory of gravity
is one of the biggest open problems of fundamental
physics. There are several approaches to this challenging issue.
Tensor models belong to the promising candidates to understand
quantum gravity (QG) in dimension $D\geq 3$
\cite{Rivasseau:2014ima}-\cite{Rivasseau:2011hm}. Tensor models come
from group field theory, which is a second-quantization of
the loop quantum gravity, spin foam and certainly from matrix models \cite{Oriti:2009wn}. 
Tensorial group
field theory (TGFT) is quantum field theory (QFT) over group manifolds. 
It can also be viewed as a new proposal for quantum field theories
based on a Feynman path integral, which generates random graphs describing
simplicial pseudo manifolds.  

A few years ago, Razvan Gur\u{a}u \cite{Gurau:2011xq}-\cite{Gurau:2013pca}
achieved a breakthrough for this program by discovering the
generalization of t'Hooft's $1/N$-expansion \cite{Di
  Francesco:1993nw}-\cite{'tHooft:1982tz}. This allows to understand
statistical physics properties such as continuum limit, phase
transitions and critical exponents (see
\cite{Gurau:2012ix}-\cite{Gurau:2011tj} for more detail).

All field theories must be physically justified by
renormalizability. In the case of tensor models, by modifying the
propagator using radiative corrections of the form $1/{\bf p}^2$
\cite{Geloun:2011cy}, this question has been solved under specific
prescriptions \cite{BenGeloun:2011rc}-\cite{Carrozza:2013wda}. The
$\beta$-functions of such models are also derived. It has been shown
that asymptotic freedom is the generic feature of all TGFT models
\cite{BenGeloun:2012pu} and
\cite{BenGeloun:2012yk}-\cite{Carrozza:2014rba}.

Recently, important progress was made in the case of independent
identically distributed (iid) tensor models. The correlation functions
are solved analytically in the large $N$-limit, in which the dominant
graphs are called ``{\it melon}'' \cite{Bonzom:2011zz}.  This model
corresponds to dynamical triangulations in three and higher
dimensions.  The susceptibility exponent  is computed and
the model is reminiscent of certain models of branched polymers
\cite{Benedetti:2011nn}. In the continuum limit, the models exhibit two
phase transitions. 
Despite all these aesthetic results, the critical behavior of the
large-$N$ limit of the renormalizable models ({\it the melonic
  approximation}) is not yet explored. The phase transitions must be
computed explicitly. This glimpse needs to be taking into account for
the future development of the renormalizable TGFT program.

This paper extends previous work on Schwinger-Dyson equations for
matrix and tensor models. The original motivation for this method was
the construction of the $\phi^4_4$-model on noncommutative Moyal
space.  The model is perturbatively renormalizable
\cite{Grosse:2004yu}-\cite{Rivasseau:2005bh} and asymptotically safe
in the UV regime \cite{Grosse:2004by}-\cite{Disertori:2006nq}. The key
step of the asymptotic safety proof \cite{Disertori:2006nq} was
extended in \cite{Grosse:2009pa} to obtain a closed equation for the
two-point function of the model. This equation was reduced in
\cite{Grosse:2012uv} to a fixed point problem for which existence of a
solution was proved. All higher correlation functions were expressed
in terms of the fixed point solution. In \cite{Grosse:2014lxa} the
fixed point problem was numerically studied. This gave evidence for 
phase transitions and for reflection positivity of the 
Schwinger two-point function.

The noncommutative $\phi^4_4$-model solved in \cite{Grosse:2012uv} can
be viewed as the quartic cousin of the Kontsevich model which is
relevant for two-dimensional quantum gravity. This leads immediately
to the question to extend the techniques of 
 \cite{Grosse:2009pa, Grosse:2012uv} to tensor models of rank $D\geq 3$.
In \cite{Samary:2014tja} one of us addressed the closed equation for
 correlation functions of rank 3 and 4 just renormalizable TGFT. The
 two-point functions are given perturbatively using the iteration
 method. The main challenge in this new direction is to perform the
 combinatorics of Feynman graphs and to solve the nontrivial integral
 equations of the correlators.  The nonperturbative study of all
 correlation functions need to be investigated carefully.  

 In this paper we push further this program.  For this, we consider
 the just renormalizable tensor model of the form $\vp^4_5$ without
 gauge condition, whose dynamics is described by the propagator of
 the form $1/{\bf p}^2$.  In the melonic approximation, the
 Schwinger-Dyson equations are given. The closed equation of the
 two-point and four-point functions are derived and its solution discussed.

 The paper is organized as follows. In section \ref{sec2}, proceeding
 from the definition of the model and its symmetries, we give the
 Ward-Takahashi identities which result from these symmetries. In
 section \ref{sec3} we find the melonic approximation of the
 Schwinger-Dyson equation. Section \ref{sec4} investigates the closed
 equation for two- and four-point functions. 
 Section \ref{sec5} is devoted to the closed equation of the four-point correlation functions. 
In section \ref{sec6} we
 solve the equation obtained.  In Section
 \ref{sec7} we give the conclusions, open questions and future work.

\noindent
\section{The Models}\label{sec2}
The model we will be mainly considering here is a tensorial
$\phi^{4}$-theory on $\text{U}(1)^{\times 5}$. Namely,
\begin{eqnarray}\label{action}
    S[\bvp, \vp]&=&\int_{\text{U}(1)^{5}} d{\bf g}
\,\bar \vp({\bf g})(-\Delta+m^{2})\vp({\bf g})\cr
  &+&
\frac{\lambda}{2}\sum_{c=1}^{5}\int_{\text{U}(1)^{20}}
d{\bf g} \,d{\bf g}' \,d{\bf h} \,d{\bf h}' \;
\bar \vp({\bf g})\vp({\bf g'})\bar \vp({\bf h})
\vp({\bf h'})K_{c}({\bf g},{\bf g'},{\bf h},{\bf h'})\,,
\label{eq-actionphi4}
\end{eqnarray}
where $\Delta=\sum_{\ell=1}^{5}\Delta_{\ell}$ and $\Delta_{\ell}$ is
the Laplace-Beltrami operator on ${U}(1)$ acting on colour-$\ell$
indices \cite{Carrozza:2012uv}, bold variables stand for
$5$-dimensional variables (${\bf g}=({\rm g_{1},\dots,g_{5})}$), and $K_{c}$
identifies group variables according to a vertex of colour 
$c\in\{1,2,\dots,5\}$. Figure~\ref{f-vertex} shows the vertex of colour 1.
\begin{figure}[!htp]
  \centering
\subfloat[A vertex of colour $1$]{{\label{f-vertex1}}\includegraphics[scale=1]{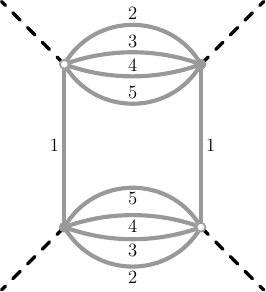}}
  \hspace{4cm}
 \subfloat [A generic vertex]{{\label{f-vertexlight}}\includegraphics[scale=1]{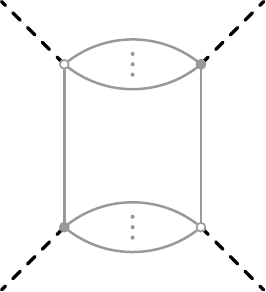}}
  \caption{Vertices}
  \label{f-vertex}
\end{figure}

The statistical physics
description of the model is encoded in the partition function:
\beq
\mathcal{Z}[\bar J,J]=\int\, D\vp D\bvp 
e^{-S[\bvp, \vp]+J\bvp+\vp \bar J}=e^{W[\bar J, J]},
\eeq
where $\bar J$ and $J$ represent the sources and $W[\bar J,J]$ is the
generating functional for the connected Green's functions. 
Then the $N$-point Green functions take the form
\bea
G_N({\bf g}_1,\cdots, {\bf g}_{2N})
=\frac{\partial \mathcal Z(\bar J, J)}{\partial J_1\partial \bar J_1
\cdots \partial J_N\partial \bar J_N}\Big|_{{\bf J}={\bf \bar J}=0}.
\eea
Now let $\vp_{class}$ denote the classical field defined by 
the expectation value of $\vp$ in the presence of sources $J$,$\bar J$: 
\beq
\vp_{class}=\langle \vp \rangle=\frac{\delta W[\bar J, J]}{\delta \bar J},
\quad  \bvp_{class}=\langle \bvp \rangle=\frac{\delta W[\bar J, J]}{\delta  J}.
\eeq
Then the 1PI effective action $\Gamma_{1PI}$ is given by 
the Legendre transform of $W[\bar J, J]$ as
\beq
\Gamma_{1PI}=-W[\bar J,J]+\int (J\bvp_{class}+\vp_{class} \bar J).
\eeq 
\noindent

The correlation functions can be computed perturbatively by expanding
the interaction part of the action
\eqref{eq-actionphi4}:
\begin{align}
\label{partition}
G_N({\bf g}_1,\cdots, {\bf g}_{2N}) \sim
\sum_{n=0}^\infty & \frac{(-\lambda)^n}{2^nn!}
\int\, d\mu_C\,
\bar{\varphi}({\bf g}_1) \cdots 
\varphi({\bf g}_{2N}) 
\\[-1ex]
&\times
\Big[ \sum_{c=1}^{5}\int_{\text{U}(1)^{20}}
d{\bf g} \,d{\bf g}' \,d{\bf h} \,d{\bf h}' \;
\bar \vp({\bf g})\vp({\bf g'})\bar \vp({\bf h})
\vp({\bf h'})K_{c}({\bf g},{\bf g'},{\bf h},{\bf h'})\Big]^n
\nonumber
\end{align}
where $d\mu_C$ is the Gaussian measure with covariance $C$ i.e:
\bea
\int\, d\mu_C\, \bvp({\bf g})\vp({\bf g'})=C({\bf g}, {\bf g'}),\quad 
\int\, d\mu_C\, \vp({\bf g})\vp({\bf g'})
=\int\, d\mu_C\, \bvp({\bf g})\bvp({\bf g'})=0.
\eea

In all of this paper we consider the Fourier transform of the field
$\vp$ to `momentum space' and write $\vp(p_1,\cdots,
p_5)=\vp_{12345}=\vp_{{\bf p}}$, with ${\bf p}\in \mathbb{Z}^5$. We
define a unitary transformation of rank-$D$ tensor fields $\vp$,
$\bvp$ under the tensor product of $D$ fundamental representations of
the unitary group $\mathcal U_{\otimes}^{N_D}:=\otimes_{i=1}^D
U(N_i)$.  For $U^{(a)}\in U(N_a)$, $a=1,2,\cdots,D$, we define
\bea
\vp_{12\cdots D}\rightarrow [U^{(a)}\vp]_{12\cdots a\cdots D}
=\sum_{p'_a\in Z} U^{(a)}_{ p_ap'_a}\vp_{12\cdots  a'\cdots D},
\\
\bvp_{12\cdots D}\rightarrow [\bvp U^{\dag(a)}]_{12\cdots a\cdots
  D}=\sum_{p'_a\in Z} \bar U^{(a)}_{p_ap'_a}\bvp_{12\cdots a'\cdots
  D}.  
\eea 
Here, $p'_a$ or simply $a'$ is the momentum index at the
position $a$ in the expression $\vp_{12\cdots a'\cdots D}$.  For
$N_i=N$, we choose the interaction terms of \eqref{action} in such a way that
they are invariant under the transformation $U^{(a)}$, i.e.\  $
\delta^{(a)} S^{\inter}=0.  $ Note that the measure $d\vp d\bvp$ is
also invariant under $U^{(a)}$.  Let us consider now the infinitesimal
Hermitian operator corresponding to the generator of unitary group
$U(N_a)$, i.e.\  
\bea 
U_{pp'}^{(a)}=\delta_{pp'}^{(a)}+iB_{pp'}^{(a)}
+O(B^2),\quad
\bar{U}_{pp'}^{(a)}=\delta_{pp'}^{(a)}-i\bar{B}_{pp'}^{(a)}+O(\bar
B^2), 
\eea 
with $\bar{B}_{pp'}^{(a)}={B}_{p'p}^{(a)}$. Then the
variation of the partition function respect to $B$, i.e.  $
\frac{\delta \ln \mathcal Z}{\delta B}=0$ gives the Ward-Takahashi
identities which are written as 
\bea\label{Wid}
&&\sum_{p_2,\cdots,p_D}\big(C^{-1}_{m2\cdots D}-C^{-1}_{n2\cdots
  D}\big)\langle \vp_{[\alpha]}\bvp_{[\beta]}\vp_{n2\cdots D}\bvp_{
  m2\cdots D}\rangle_c =\delta_{m\alpha_1}\langle
\vp_{n\alpha_2\cdots\alpha_D}\bvp_{\beta_1\cdots\beta_D} \rangle_c\cr
&&\quad\quad\quad\quad\quad\quad
-\delta_{n\beta_1}\langle\bvp_{m\beta_2\cdots\beta_D}
\vp_{\alpha_1\cdots\alpha_D}\rangle_c, 
\eea 
where $C_{p_1\cdots p_D}$ denotes the 
 propagator. For more
detail concerning relation \eqref{Wid} see \cite{Samary:2014tja}. The
correlation functions with insertion of strands are denoted by
$G^{\mathrm{ins}}_{[mn]\cdots}=\vp_{[\alpha]}\bvp_{[\beta]}\vp_{n2\cdots
  D}\bvp_{ m2\cdots D}$. Then the relation \eqref{Wid} takes form as
\bea \label{inser}
\sum_{2,3,\cdots, D}\big(C^{-1}_{m2\cdots D}-C^{-1}_{n2\cdots
  D}\big)G^{\mathrm{ins}}_{[mn]\cdots}=G_{n\cdots}-G_{m\cdots}.  
\eea

The model \eqref{eq-actionphi4} is (just) renormalizable to all orders
of perturbation theory. See Refs
\cite{BenGeloun:2011rc}-\cite{Carrozza:2012uv} for more detail.

\section{Schwinger-Dyson equation in the melonic approximation}\label{sec3}

We start by writing the Schwinger-Dyson equations for the one-particle
irreducible $2$- and $4$-point functions of the model
\eqref{eq-actionphi4}. We use the following graphical conventions:
dashed lines symbolize amputated external legs, a black circle represents a
connected function whereas two concentric circles stand for a
one-particle irreducible function. Finally, in order to lighten
equations, we will use the generic vertex of \cref{f-vertexlight} to
mean the sum of the five different coloured interactions. Note that 
\eqref{eq-SD2pt} has been derived in  \cite{Samary:2014tja}.
\begin{align}
  \raisebox{-.5\height}{\includegraphics[scale=.5]{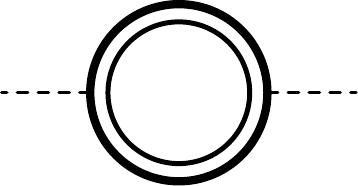}}\hspace{.5cm}
  =\Sigma_{\bf a}=& \hspace{.5cm}
  \raisebox{-.25\height}{\includegraphics[scale=.5]{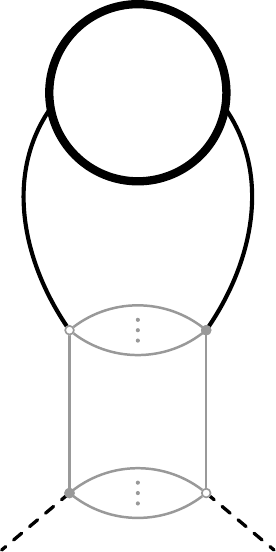}}\quad +\quad
  \raisebox{-.5\height}{\includegraphics[scale=.5]{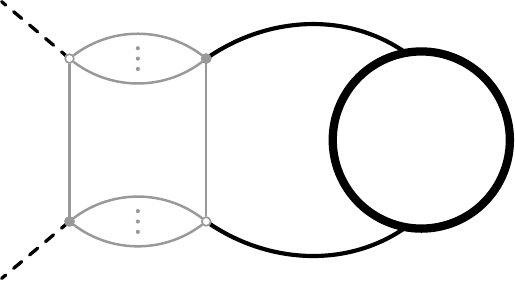}}\quad
  +\quad
  \raisebox{-.5\height}{\includegraphics[scale=.5]{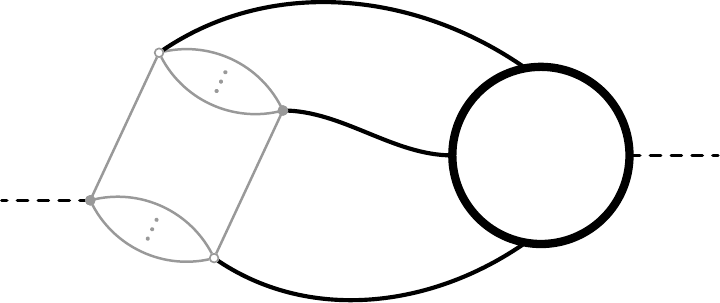}}\ .\label{eq-SD2pt}\\
  \nonumber\\
  \nonumber\\
  \raisebox{-.5\height}{\includegraphics[scale=.5]{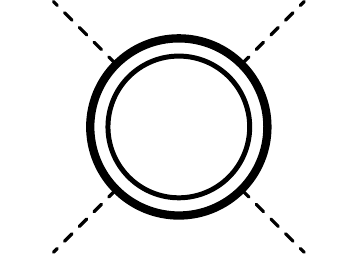}}\hspace{.5cm}
  =\Gamma_{\bf a}=& \hspace{.5cm}
  \raisebox{-.5\height}{\includegraphics[scale=.5]{schwinger-8}}\quad
  +\quad
  \raisebox{-.25\height}{\includegraphics[scale=.5]{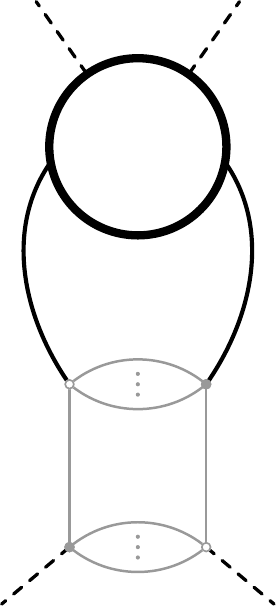}}\quad +\quad
  \raisebox{-.5\height}{\includegraphics[scale=.5]{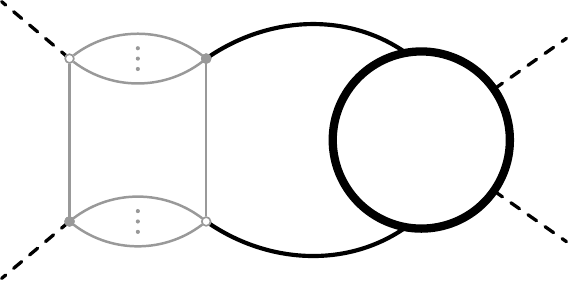}}\nonumber\\
  \nonumber\\
  &\hspace{6cm}+\quad
  \raisebox{-.5\height}{\includegraphics[scale=.5]{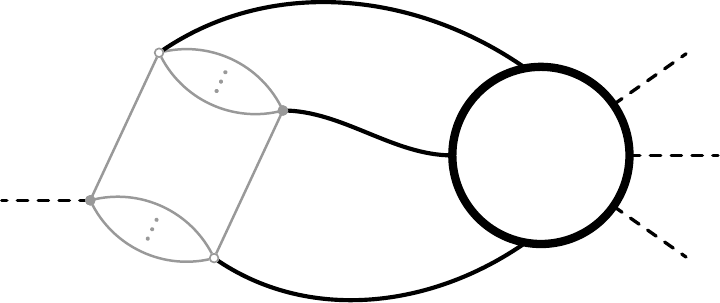}}\ .\label{eq-SD4pt}
\end{align}
We now want to restrict our attention to the melonic part of the $2$-
and $4$-point functions. Let $\cG$ be a $2$- or $4$-point Feynman
graph of model \eqref{eq-actionphi4}. Let denoted by $\omega(\cG)$ the degree of the tensor graph $\cG$ i.e:
\bea
\omega(\cG)=\sum_{\mbox{ $J$ jacket of $\cG$}} {\rm g}_J
\eea
where ${\rm g}_J$ is the genus of the jacket $J$. We impose 
$\omega(\cG)=0$. We will
prove that not all terms of \cref{eq-SD2pt,eq-SD4pt} contribute to
the melonic functions. A simple way of computing the degree $\omega$ of a
graph is to count its number $F$ of faces. Indeed, the two are related in
the following way (in dimension $5$, for a degree $4$ interaction)\cite{Samary:2014tja}:
\begin{equation}
  \label{eq-numberfaces}
  F=4V+4-2N-\frac{1}{12}(\widetilde\omega(\cG)-\omega(\partial\cG))-(C_{\partial\cG}-1)
\end{equation}
where  $V$ 
is the number of vertices of $\cG$, $N$ its number of
external legs, and $C_{\partial\cG}$ is the number of connected
components of its boundary graph $\partial\cG$ and 
$
\widetilde{\omega}(\cG)=\sum_{\widetilde J\subset \cG} {\rm g}_{\widetilde J}
$
with $\widetilde{J}$ the pinched jacket associated with a jacket $J$ of $\cG$. Recall that the Feynman graphs
here are so-called uncoloured graphs and, as a consequence, a face is a
cycle of colours $0i$, $i\in\set{1,2,\dots,5}$ \cite{Bonzom:2012hw}. A
detailed analysis of coloured graphs \cite{BenGeloun:2012pu,Samary:2014tja}  allows to prove that
$F(\cG)=\Fm(\cG)= 4V+4-2N$, if and only if   $\widetilde
\omega(\cG)=\omega(\partial\cG)=C_{\partial\cG}-1=0$. Moreover
$F\les\Fm$. We can thus prove the following

\begin{lemma}
  \label{thm-MelonicSD}
  The Schwinger-Dyson equations for the \emph{melonic} $2$- and $4$-point
  functions of model \eqref{eq-actionphi4} are ($\textup{m}$ stands for melonic):
  \begin{align}
  \raisebox{-.5\height}{\includegraphics[scale=.5]{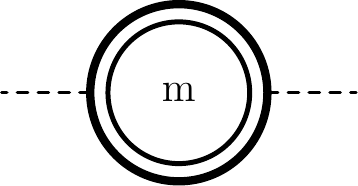}}\hspace{.5cm}
  =& \hspace{.5cm}
  \raisebox{-.25\height}{\includegraphics[scale=.5]{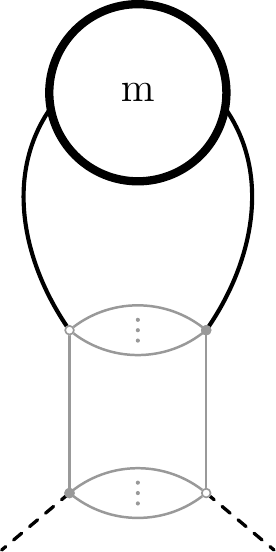}}\ ,\label{eq-SD2ptMelon}\\
  \nonumber\\
  \nonumber\\
  \raisebox{-.5\height}{\includegraphics[scale=.5]{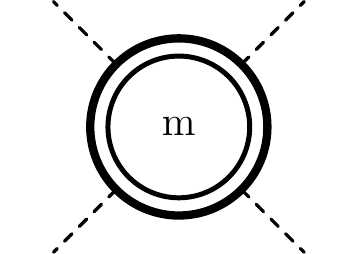}}\hspace{.5cm}
  =& \hspace{.5cm}
  \raisebox{-.5\height}{\includegraphics[scale=.5]{schwinger-8}}\quad
  +\quad
  \raisebox{-.25\height}{\includegraphics[scale=.5]{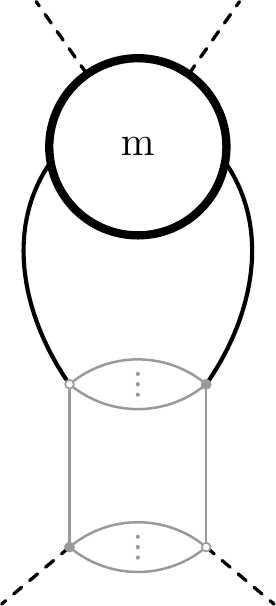}}\ .\label{eq-SD4ptMelon}
\end{align}
\end{lemma}
\begin{proof}
  The right-hand side of \cref{eq-SD2pt,eq-SD4pt} involve connected
  $2$-, $4$-, and $6$-point function insertions and a generic
  vertex. Let $\cG$ be a graph contributing to the left-hand side of 
\eqref{eq-SD2pt} or \eqref{eq-SD4pt} and let $F$ be its number of
faces. Let us study a term of the right-hand side of the equation
under consideration. The number of faces of a graph contributing to
its insertion is written $F'$. Clearly $F=F'+\delta F$ where $\delta
F\ges 0$. The additional internal faces are created by closing the
external faces of the insertion with the new edges connected to the
new vertex. As a consequence, $\delta F$ is bounded above by the
number of faces of the new vertex which do not contain its external
legs. Note also that $F\les\Fm'+\delta F$.

Let us now consider \cref{eq-SD2pt} and the lying tadpole of its
right-hand side (second term). In this case, $\delta F\les 1$. From
\cref{eq-numberfaces}, $\Fm'=4V'$ ($V'$ being the number of
vertices of the connected $2$-point insertion) and $F\les
4V'+1<4(V'+1)=\Fm$. Thus whatever the insertion, the graph $\cG$
cannot be melonic. The same type of argument holds for the other terms
but for the sake of clarity, let us repeat it for the last term of
\cref{eq-SD4pt}. Here $\delta F\les 5$ and $\Fm'=4V'-8$. Their sum
never reaches $\Fm=4V-4=4V'$.\\

The only terms which survive this analysis are the first one of
\cref{eq-SD2pt}, and the first and second ones of
\cref{eq-SD4pt}. Moreover it also proves that for a graph to be
melonic, the corresponding insertion needs to be melonic too. Note
that a melonic graph necessarily has a melonic boundary
\cite{Geloun:2012fq, Samary:2012bw}. Finally, such arguments
also fix the orientation, and the colour, of the boundary graph of the
$4$-point insertion in the second term on the right-hand side of
\ref{eq-SD4ptMelon}, see fig.\ \ref{f-bdry4pt} for a zoom into this term.
\end{proof}
\begin{figure}[!htp]
  \centering
  \includegraphics[scale=.8]{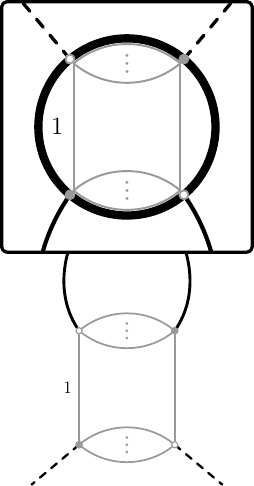}
  \caption{Boundary structure of a melonic $4$-point insertion}
  \label{f-bdry4pt}
\end{figure}

Note that the Schwinger-Dyson equation \eqref{eq-SD2ptMelon} and
\eqref{eq-SD4ptMelon} are easy to describe. Taking into account
\eqref{eq-SD2ptMelon} we do not need to write the Ward-Takahashi
identities before getting the closed equation of the two-point
functions.

\section{Two-point correlation functions}\label{sec4}

We now want to use the melonic approximation to obtain a closed equation
for the 1PI two-point function $\Sigma_{a_1,\ldots,a_5}$. 
For sake of simplicity write $\mathbf{a}=(a_1,\ldots,a_5)\in\mathbb{Z}^5$. 
Setting each constant $\lambda_\rho$  ($\rho=1,\ldots,5$) equal to the
bare coupling constant, $\lambda_\rho=\lambda$, we can express the 
1PI two-point function $\Sigma_{{\bf a}}$ in terms of the 
renormalized quantities by using the Taylor expansion 
\begin{eqnarray}
\label{one}
\Sigma_{{\bf a}}&=&\Sigma_{\bf 0}+|{\bf a}|^2\frac{\partial \Sigma_{{\bf a}}}{\partial |\mathbf{a}|^2}\Big|_{{\bf a}=\mathbf{0}}+
\Sigma_{{\bf a}}^{\mtr{r}}\cr
&=&(Z-1)|{\bf a}|^2+Zm^2-m_{\mtr{r}}^2+
\Sigma_{{\bf a}}^{\mtr{r}},
\end{eqnarray}
with
\begin{eqnarray}
m^2=\frac{m_{\mtr{r}}^2+\Sigma_{\bf 0}}{Z},\quad Z=1+\frac{\partial \Sigma_{\bf a}}{\partial |\mathbf{a}|^2}\Big|_{{\bf a}={\bf 0}}.
\end{eqnarray}

Moreover the following renormalization conditions 
\beq\label{con}
\Sigma^{\mtr{r}}_{{\bf 0}}=0,\quad \frac{\partial \Sigma^{\mtr{r}}_{{\bf a}}}{\partial a_\rho^2}\Big|_{{\bf a}={\bf 0}}=0
\eeq
hold. \\

The propagator $C$, given explicitly by $C^{-1}_{{\bf p}}=Z(|{\bf
  p}|^2+m^2)$, is related to the dressed propagator $G_{\mathbf a}$ by
means of the the Dyson relation $G^{-1}_{\bf a}=C^{-1}_{\bf
  a}-\Sigma_{\bf a}$. Then using the Schwinger-Dyson equations for
$\Sigma_{\bf a}$, given in \eqref{eq-SD2ptMelon}, we get
\begin{align}
\nonumber 
\Sigma_{{\bf a}}=-Z^2\lambda\sum_{{p_1,p_2,p_3,p_4}}^{\bf\Lambda}
&\Big[\frac{1}{C^{-1}_{a_1{p_1p_2p_3p_4}}-\Sigma_{a_1{p_1p_2p_3p_4}}}+
\frac{1}{C^{-1}_{p_1a_2p_2p_3p_4}-\Sigma_{p_1a_2p_2p_3p_4}} \\ \nonumber 
+&\,\,\frac{1}{C^{-1}_{p_1p_2a_3p_3p_4}-\Sigma_{p_1p_2a_3p_3p_4}}+
\frac{1}{C^{-1}_{p_1p_2p_3a_4p_4}-\Sigma_{p_1p_2p_3a_4p_4}} \\
\label{fab} +&\,\,\frac{1}{C^{-1}_{{p_1p_2p_3p_4}a_5}
-\Sigma_{{p_1p_2p_3p_4}a_5}}\Big] \,.
\end{align}
The sums are performed over the integers $p_i\in \mathbb Z$ with some cutoff 
${\bf\Lambda}$. For $\rho=1,\ldots,5$, 
let $\sigma_{\rho}$ be the action of $\mathfrak{S}_5$ which permutes 
the strands with momenta $\mathbf{p}$ as follows: 
\begin{align*}
 \sigma_1(p_1p_2p_3p_4p_5) & = ({p_2p_1p_3p_4}p_5), \\
 \sigma_2(p_1p_2p_3p_4p_5) & = ({ p_2 p_3p_1p_4p_5}), \\
                           & \,\,\, \vdots \\
 \sigma_4(p_1p_2p_3p_4p_5) &  = (p_2p_3p_4p_5p_1),
\end{align*}
and $\sigma_5$ trivially. Notice that the value of the propagator
$C_{\mathbf{p}}$ remains invariant under the action of all these
$\sigma_\rho$, $C_{\mathbf{\sigma_\rho(p)}}=C_{\mathbf{p}}$.  After
combining \eqref{one} and \eqref{fab} we can obtain, by using
\begin{equation}
C^{-1}_{a_1{p_1p_2p_3p_4}}-\Sigma_{a_1{p_1p_2p_3p_4}}
=a_1^2+\sum_{i=1}^4p_i^2+m^2_{\mtr r}-\Sigma_{a_1{p_1p_2p_3p_4}}^{\mtr r}, 
\label{prop.explicit}
\end{equation}
the expression

\begin{align}
 (Z-1)|{\bf a}|^2+Zm^2-m_{\mtr{r}}^2+ \Sigma_{\mathbf{a}}^{\mathrm{r}}
=-Z^2\lambda\sum_{\rho=1}^5\sum_{{p_1p_2p_3p_4}}^{\bf \Lambda}\frac{1}{(a_\rho^2+\sum_{i=1}^4p_i^2)+m_{\mtr r}^2
-\Sigma_{\sigma_\rho(a_\rho{p_1p_2p_3p_4})}^{\mtr{r}}}\label{closed}.
\end{align}
We now can evaluate at  $\mathbf{a}=\mathbf 0$ to get rid of the term $Zm^2-m_{\mtr{r}}^2$, 
which according to this equation is given by
\begin{align}
\label{varia}
Zm^2-m_{\mtr{r}}^2=-Z^2\lambda\sum_{{p_1p_2p_3p_4}}^{\bf \Lambda}\sum_{\rho=1}^5\frac{1}{\sum_{i=1}^4p_i^2+m_{\mtr r}^2-
\Sigma_{\sigma_\rho(0{p_1p_2p_3p_4})}^{\mtr{r}}}.
\end{align}
Replacing the expression \eqref{varia} in \eqref{closed}, we obtain 
\begin{align}
(Z-1)|{\bf a}|^2+\Sigma_{{\bf a}}^{\mtr{r}}
=-Z^2\lambda\sum_{{p_1p_2p_3p_4}}^{\bf \Lambda}\sum_{\rho=1}^5&
\Big[\frac{1}{a_\rho^2+|{\bf p}|^2+m_{\mtr r}^2-
\Sigma_{\sigma_\rho(a_\rho{p_1p_2p_3p_4})}^{\mtr{r}}}\cr
&-\frac{1}{|{\bf p}|^2+m_{\mtr r}^2-
\Sigma_{\sigma_\rho(0{p_1p_2p_3p_4})}^{\mtr{r}}}\Big].\label{eqcl}
\end{align}
Here we have defined $|{\bf p}|^2 := 
\sum_{i=1}^4p_i^2$, with some abuse of notation.
The evaluation at $\mathbf a=\sigma_\rho(a_\rho0000)$, namely 
\begin{align}
 \label{nnn}
(Z-1)a_\rho^2+\Sigma_{\sigma_\rho(a_\rho0000)}^{\mtr{r}}
=-Z^2\lambda\sum_{{\bf p}\in\mathbb{Z}^4}^{\bf \Lambda}
& \Big[\frac{1}{a_\rho^2+|{\bf p}|^2+m_{\mtr r}^2-
\Sigma_{\sigma_\rho(a_\rho p_1p_2p_3p_4)}^{\mtr{r}}} \\
& -\frac{1}{|{\bf p}|^2+m_{\mtr r}^2-
\Sigma_{\sigma_\rho(0 p_1p_2p_3p_4)}^{\mtr{r}}}\Big],
\nonumber
\end{align}
leads to 
a splitting of the renormalized 1PI two-point function as

\begin{equation}
\label{sym}
\Sigma_{a_1a_2a_3a_4a_5}^{\mtr{r}}=\sum_{\rho=1}^5\Sigma_{\sigma_\rho( a_\rho0000)}^{\mtr{r}}
\end{equation}
as a mere consequence of summing eq.\ \eqref{nnn} over $\rho=1,\ldots,
5$ and then comparing the rhs of the resulting equation with that of
eq.\ \eqref{eqcl}.
The wave function renormalization constant $Z$ can be obtained 
from differentiating \eqref{nnn} with respect to any $a_\rho^2$ 
and the subsequent evaluation at $a_\rho=0$:  
\beq\label{DineCarlos}
Z-1 =Z^2\lambda\Big[\sum_{{\bf p}\in\mathbb{Z}^4}^{\bf \Lambda}
\frac{1}{(|{\bf p}|^2+m_{\mtr r}^2-\Sigma_{0{\bf p}}^{\mtr r})^2}\Big],
\quad {\bf \Lambda}\in\mathbb{Z}^4.
\eeq
Here (\ref{con}) has been used. 
Insertion of this value for $(Z-1)$ into eq.\ \eqref{nnn} renders, 
setting $\tilde\lambda=Z^2\lambda$ and using \eqref{sym} again,
\beq\label{closeddine}
\Sigma_{a{\bf 0}}^{\mtr{r}}=
-\tilde\lambda\sum_{{\bf p}\in\mathbb{Z}^4}^{\bf \Lambda}
\Big[
\frac{1}{a^2+|{\bf p}|^2+m_{\mtr r}^2
-\Sigma_{a{\bf 0}}^{\mtr{r}}-\Sigma_{0{\bf p}}^{\mtr{r}}}
+\frac{a^2}{(|{\bf p}|^2+m_{\mtr r}^2-\Sigma_{0{\bf p}}^{\mtr r})^2}
-\frac{1}{|{\bf p}|^2+m_{\mtr r}^2-
\Sigma_{0{\bf p}}^{\mtr{r}}}\Big].
\eeq

The above equation could lead to a divergence in the limit where 
${\bf \Lambda}\rightarrow \infty$ which should compensate 
with a divergence of $\tilde{\lambda}^{-1}$. We will prove this
in sec. \ref{sec5}.

We now pass to a continuum limit in which the discrete momenta $a\in
\mathbb{Z},{\bf p}\in \mathbb{Z}^4$ become continuous. We do this here
in a formal manner. A rigorous treatment should first view the
regularized dual of $U(1)^5$ as a toroidal lattice $(\mathbb{Z}/2\Lambda
\mathbb{Z})^5$, then take an appropriate scaling limit to the 5-torus
$[-\Lambda,\Lambda]^5$ with periodic boundary conditions, and finally
$\Lambda\to \infty$. These steps should give for
(\ref{closeddine}):
\begin{align}
\label{climit}
\Sigma_{a{\bf 0}}^{\mtr{r}}&=
-\tilde\lambda\int_{\mathbb{R}^4} d{\bf p}\,\Big[\frac{a^2}{(|{\bf p}|^2+m_{\mtr r}^2-\Sigma_{0{\bf p}}^{\mtr r})^2}
+\frac{1}{a^2+|{\bf p}|^2+m_{\mtr r}^2-
\Sigma_{a{\bf 0}}^{\mtr{r}}-\Sigma_{0{\bf p}}^{\mtr{r}}}
-\frac{1}{|{\bf p}|^2+m_{\mtr r}^2-
\Sigma_{0{\bf p}}^{\mtr{r}}}\Big]
\end{align}
with $d{\bf p}=dp_1dp_2dp_3dp_4$. Because of 
 (\ref{sym}), i.e.\ $\Sigma_{0{\bf  p}}^{\mtr{r}}=\sum_{i=1}^4 
\Sigma_{p_i{\bf  0}}^{\mtr{r}}$, \eqref{climit} is a closed equation for
the function $\Sigma_{a{\bf 0}}^{\mtr{r}}$. 
Using Taylor's formula we can equivalently write this equation as
\beq
\label{climit-b}
\Sigma^{\mtr r}_{a{\bf 0}}
=
-\tilde\lambda
\int_0^{a^2} dt\,(a^2-t) 
\int_{\mathbb{R}^4} d{\bf p}\,
\frac{d^2}{dt^2}\Big(
\frac{1}{m_{\mtr r}^2+t -
\Sigma^{\mtr r}_{\sqrt{t}{\bf 0}}+ \sum_{i=1}^4 (p_i^2 - \Sigma^{\mtr r}_{p_i{\bf 0}})}\Big).
\eeq

The equation (\ref{climit-b}) is the analogue of
the fixed point equation \cite[eq.\ (4.48)]{Grosse:2012uv} for the
boundary 2-point function $G_{a0}$ of the quartic matrix model: In both
situations the decisive function satisfies a non-linear integral
equation for which we can at best expect an approximative numerical
solution. Finding a suitable method, implementing it in a computer
program and and running the computation needs time. We intend to 
report results in a future publication. At the moment we have 
to limit ourselves to a perturbative investigation of this equation, see sec. \ref{sec6}.

\section{Closed equation of the 1PI four-point functions}\label{sec5}

In this section we prove that the 
coupling constant $\wl$ is finite in the $UV$ regime. 
It will be convenient to briefly discuss first the index structure of 
the four-point function. 
$\Gamma^4$ has $10$ indices: Each external coloured line of 
$\varphi_\mathbf{a}$ and $\varphi_\mathbf{b}$
should be paired with one of the complex conjugate fields 
$\bar\varphi_\mathbf{c}$ and $\bar\varphi_\mathbf{d}$
in the vertex 
$\varphi_\mathbf{a}\bar\varphi_\mathbf{c}\varphi_\mathbf{b}\bar\varphi_\mathbf{d}$. 
That is to say that $\mathbf{c}$ and $\mathbf{d}$ are expressed\footnote{More precisely, 
$\mathbf{c}=(\pi_1\circ\varrho)(\mathbf{a},\mathbf{b})$ 
and $\mathbf{d}=(\pi_2\circ\varrho)(\mathbf{a},\mathbf{b})$ 
where, $(\mathbf{a,b})\in \Z^{10}$, 
$\pi_1$ and $\pi_2$ are the projections in the first or 
second factor of $\Z^5\oplus\Z^5$, and $\varrho$ is 
a permutation in $\mathfrak{S}_{10}$ that allows 
colour conservation.}
in terms of $(\mathbf{a,b})$. For instance,
for the vertex of colour 1, represented in fig. (\ref{f-vertex1}),
$\mathbf c=(a_5a_4a_3a_2b_1)$, and $\mathbf d=(b_5b_4b_3b_2a_1)$.
The external lines for that vertex look as follows:

\begin{center}
\includegraphics{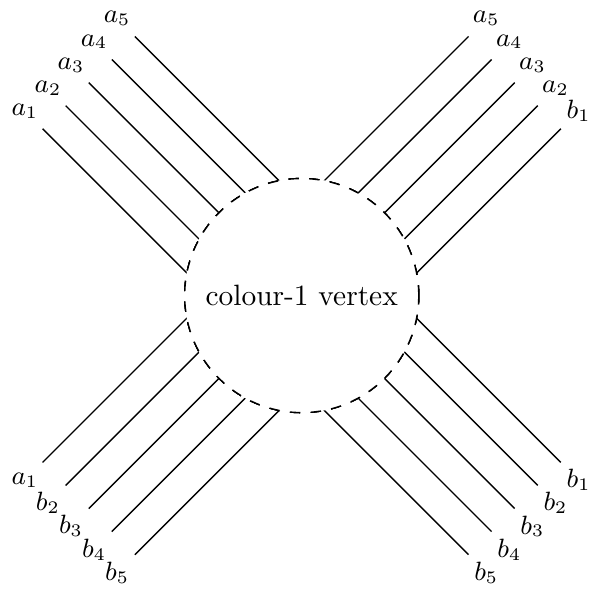}
\end{center}

We now excise the vertex in the rhs of the melonic approximation of
the Schwinger-Dyson equation for the four-point function 
and write its value, $-Z^2\lambda$, instead. The first graph 
in the rhs of eq. \eqref{eq-SD4ptMelon} is precisely the vertex.  
In the second graph, after removing the vertex, 
a jump in the colour $1$ occurs; this 
can be understood as an insertion, whose value we give now.
The removal of the colour-$1$ vertex in that graph leaves
the following graph, where the upper dotted lines have indices $\mathbf{a}=(a_1a_2a_3a_4a_5)$ and $\mathbf{c}=(a_5a_4a_3a_2b_1)$.
\begin{align}
G_{a_1a_2a_3a_4a_5}^{-1}G_{a_5a_4a_3a_2b_1}^{-1} G_{[a_1b_1]a_2a_3a_4a_5}^{\mathrm{ins}}\quad =& \quad 
  \raisebox{-.4\height}{\includegraphics{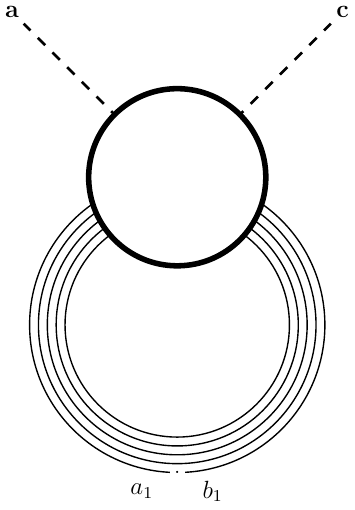}}\hspace{.5cm}.
\end{align}

According to \eqref{inser}, the value of that insertion is 
\[G_{[a_1b_1]a_2a_3a_4a_5}^{\mathrm{ins}}= \frac{1}{Z(a_1^2-b^2_1)}\left(G_{a_1a_2a_3a_4a_5}-G_{a_5a_4a_3a_2b_1} \right).\] 
In general any of the vertex in this model has a privileged colour
$i$ (i.e. the colour $i$ is with the neighbour vertically and the 
remaining colours are connected with the other neighbouring field, sidewards). 
The excised graph for the `colour $i$'-vertex has then the following value:
\[G_{a_1a_2a_3a_4a_5}^{-1} G_{a_5 \ldots b_1\hat{a}_i\ldots a_1}^{-1}G_{[a_ib_i]a_1\ldots \hat{a}_i \ldots a_5}^{\mathrm{ins}}= \frac{1}{Z(a_i^2-b^2_i)}\left[\frac{1}{G_{a_5\ldots b_i\hat{a}_i\ldots a_2a_1}}
-\frac{1}{G_{a_1a_2a_3a_4a_5}} \right], 
\] 
where $\hat a$ means omission of $\hat a_i$ (and this index is substituted by $b_i$) and, 
accordingly, $\mathbf{c}=(a_5\ldots b_i\hat{a}_i\ldots a_2a_1)$.
Then the full equation for $\Gamma^{4, \mathrm{ren}}_{a_1a_2a_3a_4a_5b_1b_2b_3b_4b_5}$ is given by the
sum over these two kinds of graphs over all the vertices of the model, to wit
\begin{align*}
\Gamma^{4, \mathrm{ren}}_{\mathbf{a,b}}
  = & \sum_{i=1}^5 -Z^2\lambda(1+G_{a_1a_2a_3a_4a_5}^{-1} G_{a_5 \ldots b_1\hat{a}_i\ldots a_1}^{-1}G_{[a_ib_i]a_1\ldots \hat{a}_i \ldots a_5}^{\mathrm{ins}})  \\
  = & -Z^2\lambda\left(5+\frac{1}{Z(a_1^2-b^2_1)}\left[\frac{1}{G_{a_5a_4a_3a_2b_1}}-\frac{1}{G_{a_1a_2a_3a_4a_5}} \right] \right. \\
  & + \frac{1}{Z(a_2^2-b^2_2)}\left[\frac{1}{G_{a_5a_4a_3b_2a_1}}-\frac{1}{G_{a_1a_2a_3a_4a_5}} \right]
  +\frac{1}{Z(a_3^2-b^2_3)}\left[\frac{1}{G_{a_5a_4b_3a_2a_1}}-\frac{1}{G_{a_1a_2a_3a_4a_5}} \right] \\
  & +\frac{1}{Z(a_4^2-b^2_4)}\left[\frac{1}{G_{a_5b_4a_3a_2a_1}}-\frac{1}{G_{a_1a_2a_3a_4a_5}} \right] 
  +  \left.  \frac{1}{Z(a_5^2-b^2_5)}\left[\frac{1}{G_{b_5a_4a_3a_2a_1}}-\frac{1}{G_{a_1a_2a_3a_4a_5}} \right]\right).
  \end{align*}
By inserting the value for $G_{\mathbf{q}}$ given by \eqref{prop.explicit}, 
and by imposing the renormalization conditions, taking the limit $\mathbf{a,b}\to 0$ one readily obtains
\bea\Gamma_{0}^{{4, \mathrm{ren}}}=-5\wl\Big(1+\frac{1}{Z}\Big).
\eea
By imposing the cutoff $\Lambda,$  we can show (perturbatively) that 
the wave function renormalization (for 
a similar computation cf. Lemma 5 in \cite{BenGeloun:2012pu}) takes the form 
\bea
Z=1 +x\lambda\log(\Lambda)+\mathcal{O}(\lambda^2),\quad x\in\mathbb{R}.
\eea Then one has
\bea-\lambda_r=\Gamma_{0}^{\mathrm{ren}}\rightarrow -5\wl.\eea

\section{Solution of the integral equation}\label{sec6}

The integral equation \eqref{climit-b} is a non-linear 
integro-partial differential equation. We therefore opt for a numerical
approach. We introduce the following dimensionless variables:
\[ \alpha \equiv \frac{a}{m_\rc},\qquad \tau\equiv \frac{t}{m^2_\rc}, \qquad  \boldsymbol{\rho} \equiv \frac{\mathbf p}{m_\rc}, \qquad \mbox{and} \quad 
\quad \gamma\equiv 1+\tau+\sum_{i=1}^4\rho_i^2,  \]
and, accordingly, we rescale the two-point function $\si(\alpha)\equiv \Sigma^\rc_{a0000}/m^2_\rc$.
Equation \eqref{climit-b} can be thus reworded:
\begin{align}
 \si(\alpha)=&-\wl\int d\boldsymbol{\rho}\int_0^{\alpha^2}d\tau(\alpha^2-\tau)
 \frac{\partial^2}{\partial \tau^2} \left\{\frac{1}{1+\tau+|\boldsymbol \rho|^2-
 \si(\sqrt{\tau},\boldsymbol \rho) } \right\}. \label{normIE}
\end{align}
 Expanding the solution in $\wl$, 
 $\si(\alpha)=\sum_{n=0}^\infty \si_n(\alpha) \wl^n$, it readily follows $\si_0(\alpha)=0$. To 
 proceed with the computation of the non-trivial orders, 
we invert the power series (in $\wl$) appearing in the 
denominator \eqref{normIE} after factoring out $\gamma$, namely  $(1-\sigma(\sqrt{\tau},\boldsymbol \rho)/\gamma)$.
First, we treat this series as a formal power series,
then we care about convergence. The idea is that in order to compute 
$\si_{n+1}$, for which we 
need $\si_{i}$, $i\leq n$, we approximate the latter functions
by near-to-`principal diagonal' Pad\'e approximants, i.e. by 
quotients of polynomials of almost equal degree; this approximation 
is valid in a certain domain and would lead to the convergence of the series
there. Shortly, a second advantage of the Pad\'e approximants will be evident.
\par	

We use the following result for the power of a series (cf.
sec. 3.5 in \cite{comtet:1974a}):  
For any $r\in \mathbb C$, the $r$-th power of a formal power
series $1+g_1t^1+\frac{1}{2!}g_2t^2+\ldots$ can be expanded
as follows:
\begin{align}
\left(1+\sum_{n\geq1} g_n\frac{t^n}{n!}\right)^r = & 
\,\, 1+ \sum_{n\geq1}\left( 
\mathbb P_n^{(r)}\frac{t^n}{n!}\right),
\end{align}
where the $\mathbb P_{n}^{(r)}$, the so-called \textit{potential polynomials}, are given in terms of the Bell polynomials
$\mathbb B_{p,q}$:
\begin{align*}
\mathbb P_n^{(r)}=&\sum_{1 \leq k\leq n} (r)_k \mathbb B_{n,k}(g_1,\ldots,g_{n-k+1}) 
\\
= & \sum_{1 \leq k\leq n} (-1)^k k! \mathbb B_{n,k}(g_1,\ldots,g_{n-k+1}).
\end{align*}

In our case, the Pochammer symbol appearing there, 
$(r)_k$, becomes $(-1)_k=(-1)^k k!$ 
As for the Bell polynomials, they are defined by
\[\mathbb B_{n,k}(x_1,\ldots,x_{n-k+1})=\sum_{c_j}\frac{n!}{c_1!c_2!\cdots c_{n-k+1}!} 
\left(\frac{x_1}{1!}\right)^{c_1}
\left(\frac{x_2}{2!}\right)^{c_2}\cdots \left(\frac{x_{n-k+1}}{{n-k+1}!}\right)^{c_{n-k+1}}.\]
The sum here runs over all the non-negative integers $c_l$ such that
the conditions 
\begin{equation}
 \sum_{i=1}^{n-k+1}c_i=k  \qquad \mbox{and}\qquad \sum_{q=1}^{n-k+1} q c_q=n 
\end{equation}
are fulfilled. It will be useful to rescale the $k$-th variable $x_k$ in 
the Bell polynomials by $x_k'= w (k!)x_k$, for a number $w\neq 0$, to
obtain a simpler expression in the lhs:

\begin{align}
\nonumber \mathbb B_{n,k}(x_1',\ldots,x_{n-k+1}')&=
\mathbb B_{n,k}\left(w (1!) x_1,w (2!)x_2,\ldots,w(n-k+1)!x_{n-k+1}\right)\\
&= w^k
\sum_{c_j}\frac{n!}{c_1!c_2!\cdots c_{n-k+1}!} 
x_1^{c_1}\cdots ( x_{n-k+1})^{c_{n-k+1}} \label{bell_simpl}.
\end{align}
\textbf{Remark.} After taking the reciprocal of the power series,
the convergence of each coefficient of $\wl^n$, $\si_n(\alpha)$, 
is not guaranteed. We denote by $\wsi_n(\alpha)$ those probably
divergent coefficients, which need to be renormalized.
Thus, taking the $(n+1)$-order in $\wl$ of $\wsi(\alpha)$, $\wsi_{n+1}(\alpha)$, 
boils down to integrate
\begin{align}\label{sattt}
\wsi_{n+1}(\alpha)=&-\int_{\mathbb R^4} d\boldsymbol\rho \int _0^{\alpha^2}d\tau(\alpha^2-\tau) 
\frac{\partial^2}{\partial \tau^2} \cr
&\qquad \left[\frac{1}{n! \gamma} \sum_{1\leq k\leq n}
(-1)^k k!\mathbb B_{n,k}\left(-1!\frac{\si_1(\zeta)}{\gamma},-2!\frac{\wsi_2(\zeta)}{\gamma},\ldots,
-(n-k+1)!\frac{\wsi_{n-k+1}(\zeta)}{\gamma}
\right)\right] \cr
=&-\int_{\mathbb R^4} d\boldsymbol\rho \int _0^{\alpha^2}d\tau(\alpha^2-\tau) 
\frac{\partial^2}{\partial \tau^2} \left[
\sum_{1\leq k \leq n}\frac{k!}{\gamma^{k+1}}\sum_{\mathbf c(k,n) } \prod_{j=1}^{n-k+1} 
\left(\frac{\wsi_j(\zeta)^{c_j}}{c_j!}\right)\right].
\end{align}
Here $\zeta=(\sqrt{\tau},\boldsymbol\rho)$ and we have made use of
\eqref{bell_simpl} with the nowhere-vanishing $w=-\gamma^{-1}$.  To obtain expressions for higher-order
solutions we use the explicit form of the Bell polynomials
\begin{align*}
  \mathbb B_{1,1}(x_1) &=x_1,  &   \mathbb{B}_{2,1}(x_1,x_2) &= x_2, & \mathbb{B}_{3,1}(x_1,x_2,x_3) &=x_3,    &  \mathbb{B}_{4,1}(x_1,x_2,x_3,x_4)&=x_4, \\
   &&   \mathbb{B}_{2,2}(x_1,x_2)&=x_1^2 ,
   &\mathbb{B}_{3,2}(x_1,x_2,x_3) &=3 x_1 x_2,  &  \mathbb{B}_{4,2}(x_1,x_2,x_3,x_4)&= 4x_1x_3+3x_2^2,\\
      &&   & 
   &\mathbb{B}_{3,3}(x_1,x_2,x_3) &= x_1^3,  &  \mathbb{B}_{4,3}(x_1,x_2,x_3,x_4)&= 6x_1^2x_2, \\
         &&   & 
   & &  &  \mathbb{B}_{4,4}(x_1,x_2,x_3,x_4)&= x_1^4. \\
   \end{align*}
The first order in perturbation theory can be given exactly ---and without using the Pad\'e approximants, nor
regularization --- and is given by 
\begin{align*}
 \si_1(\alpha)&=-2\mathrm{Vol}(\mathbb S^3)\int_0^{\alpha^2}\, d\tau\,(\alpha^2-\tau)\int_{0}^\infty d\rho \frac{\rho^2}{(1+\tau+\rho^3)^2} \\
         &=-2(2\pi^2)\int_0^{\alpha^2}\,d\tau\,\frac{\alpha^2-\tau}{4(1+\tau)} 
         =-\pi^2[(\alpha^2+1)\log(\alpha^2+1)-\alpha^2].
\end{align*}
With \eqref{sattt} in our hands, other low-order terms can be obtained:
\begin{align*}
 \wsi_0(\alpha)&=\si_0(\alpha)=  0 \\
 \wsi_1(\alpha)&= \si_1(\alpha)= -\pi^2[(\alpha^2+1)\log(\alpha^2+1)-\alpha^2]\\
 \wsi_2(\alpha)&=  -\int_{\mathbb R^4_\Lambda} d\boldsymbol\rho \int _0^{\alpha^2}d\tau(\alpha^2-\tau) 
\frac{\partial^2}{\partial \tau^2} \left(\frac{1}{\gamma^{2}} \wsi_1(\zeta)\right)\\
 \wsi_3(\alpha)&=  -\int_{\mathbb R^4_\Lambda} d\boldsymbol\rho \int _0^{\alpha^2}d\tau(\alpha^2-\tau) 
\frac{\partial^2}{\partial \tau^2} 
\left\{\frac{1}{\gamma^{3}}\left(\wsi_1^2(\zeta)+\gamma \wsi_2(\zeta)\right)\right\}\\
 \wsi_4(\alpha)&=  -\int_{\mathbb R^4_\Lambda} d\boldsymbol\rho \int _0^{\alpha^2}d\tau(\alpha^2-\tau) 
\frac{\partial^2}{\partial \tau^2} \left\{\frac{1}{\gamma^{4}}\left(\wsi_1^3(\zeta)
+2\gamma\wsi_1(\zeta) \wsi_2(\zeta)+\gamma^2\wsi_3(\zeta)\right)\right\}\\
 \wsi_5(\alpha)&=  -\int_{\mathbb R^4_\Lambda} d\boldsymbol\rho \int _0^{\alpha^2}d\tau(\alpha^2-\tau) 
\frac{\partial^2}{\partial \tau^2} \left\{\frac{1}{\gamma^{5}}\Big(\wsi_1^4(\zeta)
+3\gamma\wsi_1(\zeta)^2 \wsi_2(\zeta)\right. \\
&\qquad \qquad\qquad \qquad\qquad\qquad\quad 
\left.\big.\phantom{\frac{1}{\gamma}} + 2\gamma^2(\wsi_1(\zeta)\wsi_3(\zeta)+\wsi_2^2(\zeta))
+\gamma^3\wsi_4(\zeta)\Big)\right\}.
\end{align*}
In all these expressions $\wsi_i(\zeta)=\sum_{j=1}^4\wsi_i(p_j)+\wsi_i(\sqrt \tau)$, with 
$\zeta_0=\sqrt{\tau},\zeta_1=\rho_1,\ldots,\zeta_4=\rho_4$. Notice
that the non-linearity is evident from the third order on. \par
To shed some light on the procedure to extract the divergence 
occurring in the integral \eqref{sattt}, we consider the 
second order and then extend the method to higher orders.
The most dangerous term in 
\bea\label{express}
 \wsi_2(\alpha)
&=&-\int_{\mathbb R^4_\Lambda} d\boldsymbol\rho \int _0^{\alpha^2}d\tau(\alpha^2-\tau)\Big[\frac{6\wsi_1(\zeta)}{\gamma^4}-
\frac{4\wsi'_1(\sqrt{\tau})}{\gamma^3}+\frac{\wsi''_1(\sqrt{\tau})}{\gamma^2}\Big] 
\eea
is the last summand.
We write the Taylor expansion of $\gamma^{-2}\wsi''_1(\sqrt{\tau})$ 
at first order and get the renormalized expression $\si_2(\alpha)$ as 
\bea
\si_2(\alpha)&=&  -\int_{\mathbb R^4_\Lambda} d\boldsymbol\rho \int _0^{\alpha^2}d\tau(\alpha^2-\tau)\Big[\frac{6\si_1(\zeta)}{\gamma^4}-\frac{4\si'_1(\sqrt{\tau})}{\gamma^3}+\frac{\si''_1(\sqrt{\tau})}{\gamma^2}-\frac{\si''_1(\sqrt{\tau})}{(1+|{\boldsymbol \rho}|^2)^2}\Big] \cr
&=&-\int_{\mathbb R^4_\Lambda} d\boldsymbol\rho \int _0^{\alpha^2}d\tau(\alpha^2-\tau)\Big[\frac{6\si_1(\zeta)}{\gamma^4}-\frac{4\si'_1(\sqrt{\tau})}{\gamma^3}\Big]\cr
&&+\pi^2\int _0^{\alpha^2}d\tau(\alpha^2-\tau) \si''_1(\sqrt{\tau})\log(1+\tau) \nonumber\\
&=&-\int_{\mathbb R^4_\Lambda} d\boldsymbol\rho \int _0^{\alpha^2}d\tau(\alpha^2-\tau)\Big[\frac{6\si_1(\zeta)}{\gamma^4}-\frac{4\si'_1(\sqrt{\tau})}{\gamma^3}\Big] \nonumber\cr
&&+\pi^4\left\{(1+\alpha^2)\log(1+\alpha^2)-\alpha^2-\frac12 (1+\alpha^2)\left[\log(1+\alpha^2)\right]^2 \right\}.
\eea
The above integral is convergent and therefore $\si_2(\alpha)$ 
is well defined in the limit where $\Lambda\rightarrow \infty$.
Consider now
\bea
\wsi_{n+1}(\alpha)&=&-\int_{\mathbb R^4} d\boldsymbol\rho \int _0^{\alpha^2}d\tau(\alpha^2-\tau) 
\frac{\partial^2}{\partial \tau^2} \left[
\sum_{1\leq k \leq n}\frac{k!}{\gamma^{k+1}}\sum_{\mathbf c(k,n) } \prod_{j=1}^{n-k+1} 
\left(\frac{\wsi_j(\zeta)^{c_j}}{c_j!}\right)\right].
\eea
The integral leads to the logarithmically divergence which could be removed. We get
\bea
\si_{n+1}(\alpha)&=&-\int_{\mathbb R^4} d\boldsymbol\rho \int _0^{\alpha^2}d\tau(\alpha^2-\tau) \Bigg\{
\frac{\partial^2}{\partial \tau^2} \Big[
\sum_{1\leq k \leq n}\frac{k!}{\gamma^{k+1}}\sum_{\mathbf c(k,n) } \prod_{j=1}^{n-k+1} 
\Big(\frac{\si_j(\zeta)^{c_j}}{c_j!}\Big)\Big]\cr
&&\qquad\qquad \qquad \qquad \qquad\quad- \frac{\si''_n(\sqrt{\tau})}{(1+|{\boldsymbol \rho}|^2)^2}\Bigg\}.
\eea
The above integral is convergent in the limit where 
$\Lambda\rightarrow \infty$ using (almost) equal degree Pad\'e 
approximation. The solution of the integral equation, for small values 
of the coupling constant, is given in  \cref{bbb}, \cref{lcritical} and \cref{aaa}. 
Those plots show $\sigma(\alpha)$, computed to second order in 
$\wl$. We have used \textit{Mathematica}$^{TM}$ to obtain the 
Pad\'e approximants and to plot the solution. Their advantage
over partial Taylor sums to approximate the $\si_i$'s becomes
now clear--- those had been otherwise divergent and the only
term we introduced in order to control the divergence would not 
have been enough.

%

\begin{figure}[!htp]
  \centering
\includegraphics[scale=.87]
{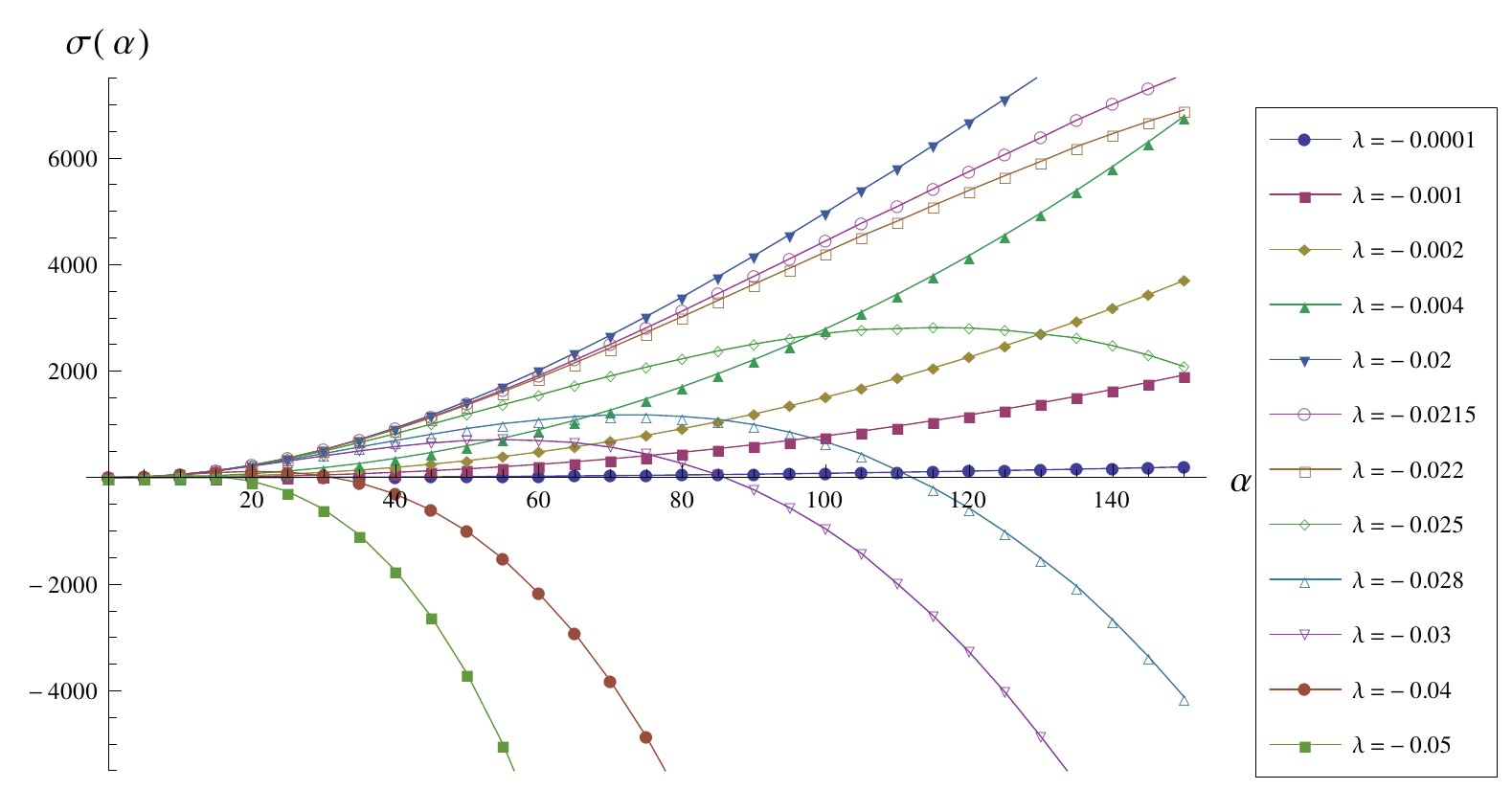}  \caption{\small 
Plot of $\si(\alpha)$ with different negative values of $\lambda$. 
 The curves are interpolations of discrete data obtained for 
 the two-point function of the $\varphi^4_5$-model (with $m_r$ set to $1$) 
 to second order in $\wl$.} \label{bbb} 
\end{figure}

\begin{figure}[!htp]
  \centering
\includegraphics[scale=1.15]
{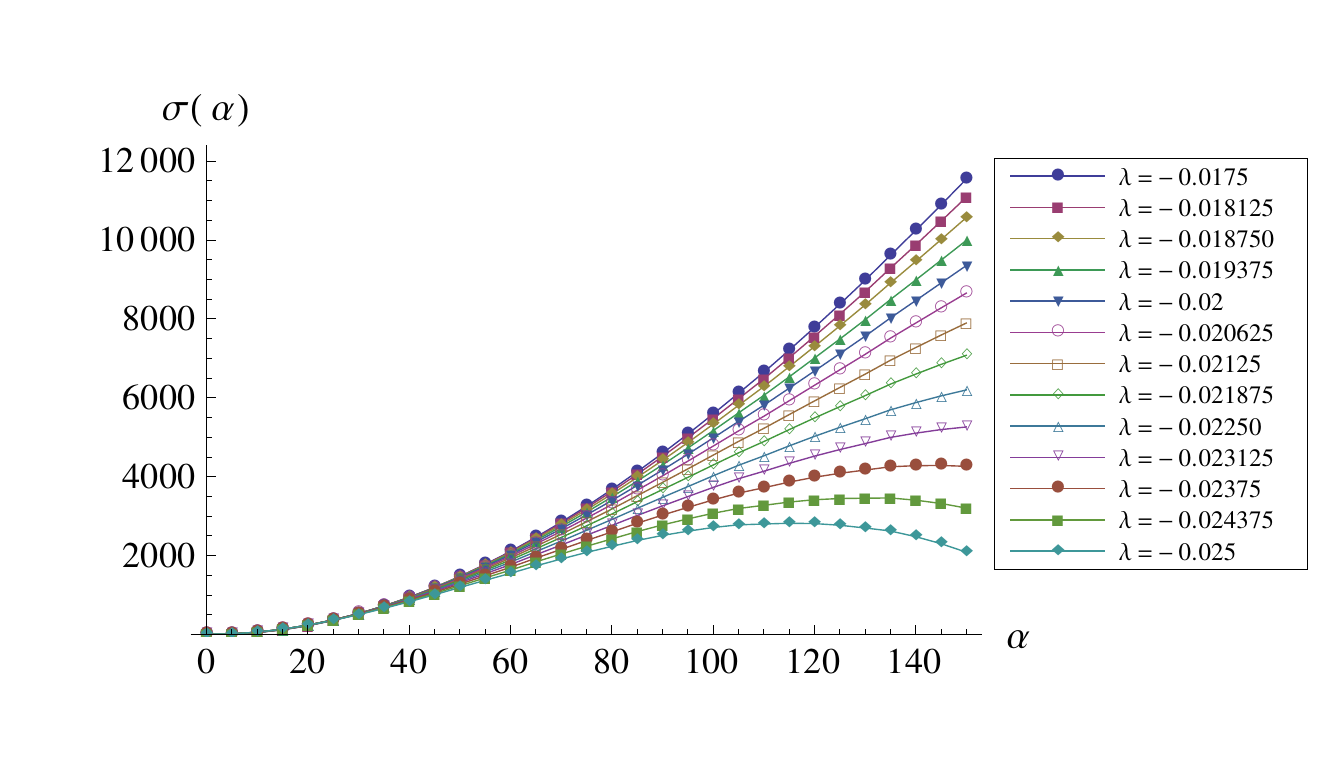}  \caption{{This is a zoom
to the region where criticality might take place. It shows 
  how the behaviour of the two-point function bifurcates
  from a certain value for the coupling constant about $\lambda \approx -0.002125 $. }
  \label{lcritical} }
\end{figure}

\begin{figure}[!htp]
  \centering
\includegraphics
{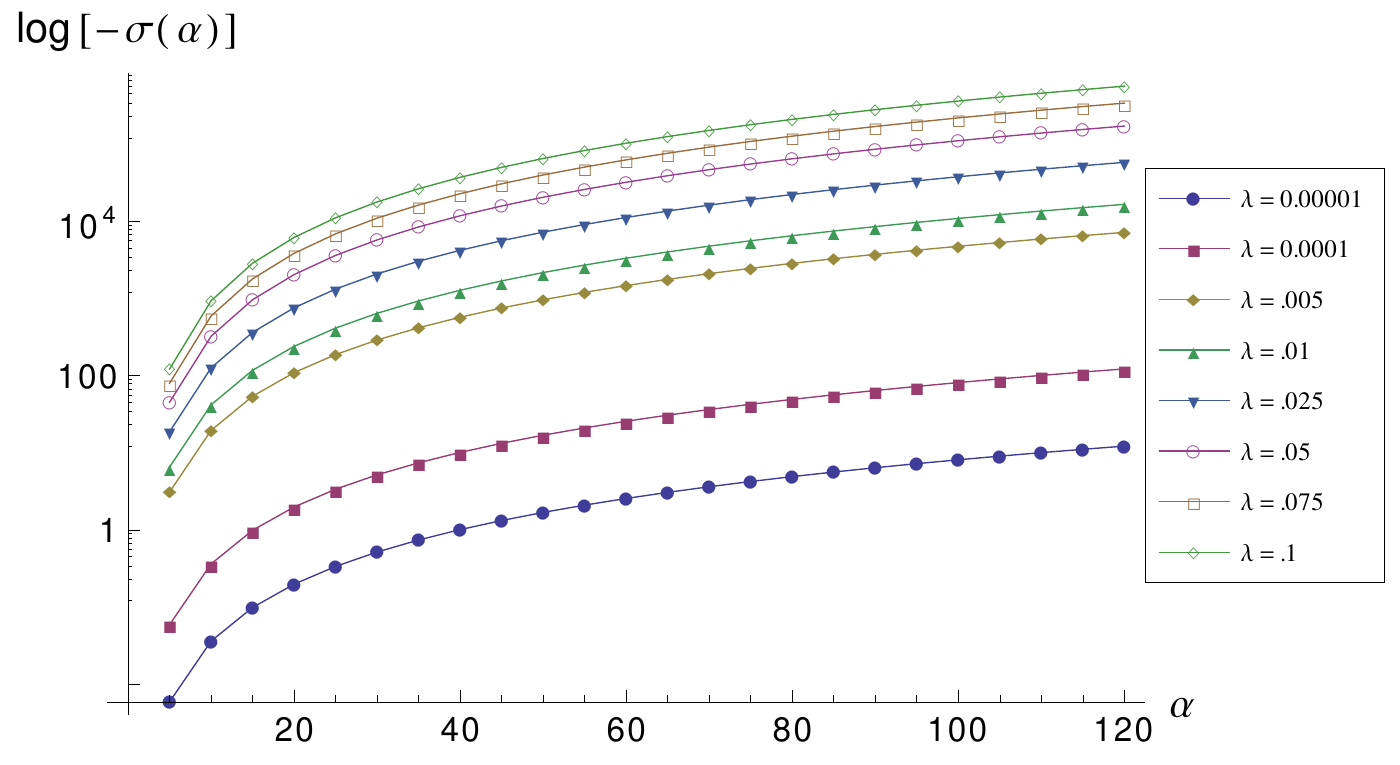}  \caption{\small 
  Plot of $\log[-\si(\alpha)]$ with different  positive values of  $\lambda$. Just as 
  in the previous plot, we interpolated a discrete graph.}
  \label{aaa}
\end{figure}


\section{Conclusion}\label{sec7}
In this paper we have considered the just renormalizable $\vp_5^4$ 
tensorial group field theory with the propagator of the form $1/{\bf p}^2$. 
We have introduced the melonic approximation of the Schwinger-Dyson equation of 
the two and four-point functions.
 This made possible, by suppressing the non-melonic graphs, to 
obtain a closed equation 
for the two-point functions. This equation is solved perturbatively. 
 It would be interesting to apply the 
melonic approximation to other tensor models
supporting a large-$N$ expansion, e.g. to multi-orientable tensor models \cite{Dartois:2014aaa}.

For future investigation remains the numerical study of the four-point function
we treated in section 5. We also plan to address the criticality of the model. 
Concretely, at certain value of the coupling constant, 
namely about $\lambda \approx -2.125\,\times 10^{-2}$, the behaviour of the two-point function
noticeably bifurcates. Thus, some criticality is promissory in 
\cref{lcritical}. To claim this we need a new, more detailed study, though; 
for instance, by solving for higher values of $\alpha$. 
The phase transitions
and the critical behaviour of the model 
could physically relevant, and in particular, interesting for applications in cosmology.
 
\section*{Acknowledgements}
The authors  thank Vincent Rivasseau for fruitful discussion.
D. Ousmane Samary's research was supported in part by Perimeter
Institute for Theoretical Physics and Fields Institute for Research in
Mathematical Sciences (Toronto).  F. Vignes-Tourneret and
R. Wulkenhaar thank the Perimeter Institute for invitation and
hospitality. Research at Perimeter Institute is
supported by the Government of Canada through Industry Canada and by
the Province of Ontario through the Ministry of Research and
Innovation. C. I. P\'erez-S\'anchez wishes to thank the DAAD (Deutscher 
Akademischer Austauschdienst), for financially supporting his 
PhD studies hitherto. We thank the referee for useful remarks,
particularly for pointing out criticality.

\end{document}